\documentclass[sn-mathphys]{sn-jnl}% Math and Physical Sciences Reference Style
\usepackage{caption}
\usepackage{lipsum}

\jyear{2021}%

\theoremstyle{thmstyleone}%
\newtheorem{theorem}{Theorem}%  meant for continuous numbers
\theoremstyle{thmstyletwo}%
\newtheorem{remark}{Remark}%
\newtheorem{lemma}{Lemma}
\theoremstyle{thmstylethree}%
\newtheorem{corollary}{Corollary}

\begin{document}

\title[Sharper Bounds on Four Lattice Constants]{Sharper Bounds on Four Lattice Constants}

%% \boldsymbol{A}uthor*[1,2]{\pfx{Dr} \fnm{Joergen W.} \spfx{van der} \sur{Ploeg} \sfx{IV} \tanm{Poet Laureate}

\author*[1]{\fnm{Jinming} \sur{Wen}}\email{jinming.wen@mail.mcgill.ca}

\author[2]{\fnm{Xiao-Wen} \sur{Chang}}\email{chang@cs.mcgill.ca}

\affil*[1]{\orgdiv{College of Information Science and Technology}, \orgname{Jinan University}, \orgaddress{ \city{Guangzhou}, \postcode{510632},
and \orgname{State Key Laboratory of Cryptology}, \city{Beijing}, \postcode{100878}, \country{China}}}

\affil[2]{\orgdiv{School of Computer Science}, \orgname{McGill University}, \orgaddress{\city{Montreal}, \postcode{H3A 0E9},  \country{Canada}}}

\abstract{The Korkine--Zolotareff (KZ) reduction, and its generalisations, are widely used lattice reduction strategies in communications and cryptography.
The KZ constant and Schnorr's constant were defined by Schnorr in 1987.
The KZ constant can be used to quantify some useful properties of KZ reduced matrices.
Schnorr's constant can be used to characterize the output quality of his block $2k$-reduction and
is used to define his semi block $2k$-reduction, which was also developed  in 1987.
Hermite's constant, which is a fundamental constant lattices, has many applications,
such as bounding the length of the shortest nonzero lattice vector and the orthogonality defect of lattices.
Rankin's constant was introduced by Rankin in 1953 as a generalization of Hermite's constant.
It plays an important role in characterizing the output quality of block-Rankin reduction, proposed by Gama et al. in 2006.
In this paper, we first develop a linear upper bound on Hermite's constant
and then use it to develop an upper bound on the KZ constant.
These upper bounds are sharper than those obtained recently by the  authors,
and the ratio of the new linear upper bound to the nonlinear upper bound,
developed by Blichfeldt in 1929, on Hermite's constant is asymptotically 1.0047.
Furthermore, we develop lower and upper bounds on Schnorr's constant.
The improvement to the lower bound over the sharpest existing one developed by Gama et al. is around 1.7 times asymptotically,
and the improvement to the upper bound over the sharpest existing one which was also developed by Gama et al. is around 4 times asymptotically.
Finally, we develop lower and upper bounds on Rankin's constant.
The improvements of the bounds over the sharpest existing ones,
also developed by Gama et al., are exponential in the parameter defining the constant.}


\footnotetext{Part of the work was presented in 2019 IEEE International Symposium on Information Theory, July 7-12, Paris, France.}
\footnotetext{This work was partially supported by
National Natural Science Foundation of China (Nos.11871248),
the Guangdong Province Universities and Colleges Pearl River Scholar
Funded Scheme (2019), Natural Science Foundation of Guangdong Province of China (2021A515010857,2022A1515010029),
Guangdong Major Project of  Basic and Applied Basic Research (2019B030302008),
and the Natural Sciences and Engineering Research
Council of Canada (NSERC) Grant RGPIN-2017-05138.}
\maketitle
\section{Introduction}
\label{s:introduction}
The lattice  generated by a full column rank matrix $\boldsymbol{A}\in \mathbb{R}^{m\times n}$
is defined by
\begin{equation}
\label{e:latticeA}
\mathcal{L}(\boldsymbol{A})=\{\boldsymbol{A}\boldsymbol{x} : \boldsymbol{x} \in \mathbb{Z}^n\}.
\end{equation}
The column vectors of $\boldsymbol{A}$ represent the basis  of $\mathcal{L}(\boldsymbol{A})$.
Here $m$ and $n$ are called the dimension and the rank of the lattice,respectively.

A matrix $\boldsymbol{Z} \in \mathbb{Z}^{n\times n}$ satisfying $\lvert\det(\boldsymbol{Z})\rvert=1$ is said to be unimodular.
For any unimodular $\boldsymbol{Z} \in \mathbb{Z}^{n\times n}$, $\mathcal{L}(\boldsymbol{A}\boldsymbol{Z})$ is the same lattice as $\mathcal{L}(\boldsymbol{A})$.
Lattice reduction is the process of finding a unimodular $\boldsymbol{Z}$ such that the column vectors of
$\boldsymbol{A}\boldsymbol{Z}$ are short and nearly orthogonal.
There are several types of lattice reduction strategies.
The Lenstra--Lenstra--Lov\'asz (LLL) reduction, the Korkine--Zolotareff (KZ) reduction and its generalisations are  popular,
and they have crucial applications in many domains including communications \cite{AgrEVZ02} and cryptography \cite{MicR08}.

The closest vector problem (CVP)
\begin{equation*}
\min_{\boldsymbol{x}\in\mathbb{Z}^n}\|\boldsymbol{y}-\boldsymbol{A}\boldsymbol{x}\|_2
\end{equation*}
arises in many applications.
To save the computational cost for solving the CVP,
LLL reduction is often used to preprocess the matrix $\boldsymbol{A}$.
In some communications applications, a number of  CVP instances with the same matrix $\boldsymbol{A}$ but different $\boldsymbol{y}$ need to be solved.
In this situation, instead of LLL reduction, KZ reduction
can be applied to preprocess $\boldsymbol{A}$.
Although it is more time consuming to perform KZ reduction than LLL reduction,
the KZ reduced matrix has better properties than the one obtained by LLL reduction,
and hence the total computational time required to solve the CVP instances may be less.
The KZ reduction also has applications in successive integer-forcing linear receiver design \cite{OrdEN13}
and integer-forcing linear receiver design \cite{SakHV13}.

It is useful to quantify the performance of KZ reduction in terms of how it shortens the lengths of basis vectors and reduces the orthogonality defect of lattice bases.
The KZ constant, defined by Schnorr \cite{Sch87}, is a measure of the quality of KZ reduced matrices.
It can be used to bound the lengths of the column vectors of KZ reduced matrices from above \cite{LagLS90,WenC18}.
Furthermore, the KZ constant has applications in bounding
the decoding radius and the proximity factors of KZ-aided successive interference cancellation
(SIC) decoders from below \cite{WenC18, Lin11, LuzSL13}.
Although the KZ constant is an important quantity, there is no formula for it.
Fortunately, it has several upper bounds which are functions of the rank $n$ of $\mathcal{L}(\boldsymbol{A})$ \cite{Sch87, HanS08, WenC18}.
Due to the importance of the  KZ constant,
an improvement to its sharpest existing upper bound presented in \cite{WenC18} is an important result in the theory of lattices.

Schnorr introduced  block and semi block $2k$-reductions \cite{Sch87}, where $k$ is an integer parameter.
Schnorr's  semi block $2k$-reduction is a relaxed form of his block $2k$-reduction,
and is a generalization of and more powerful than LLL reduction.
Furthermore, it can also be performed in polynomial time for reducing an integer matrix with suitable $k$.

A constant $\beta_k$, called  Schnorr's constant \cite{GamHKN06},
was defined in \cite{Sch87} to globally measure the drop of the diagonal entries of the R-factor of QR factorization of
block $2k$-reduced matrices.
Furthermore, it is used to define semi block $2k$-reduction.
Although $\beta_k$ is important, there is no formula for it.
Fortunately, Schnorr upper bounded it in \cite[Theorem 2.7]{Sch87}, and the upper bound was improved in \cite[Theorem 2]{GamHKN06}.
Since $\beta_k$ is used to characterize the output quality of (semi) block $2k$-reduction,
and a better upper bound on it can better quantify the quality of (semi) block $2k$-reduction,
we aim to further improve the existing sharpest upper bound in \cite[Theorem 2]{GamHKN06}.

On the other hand, Ajtai showed that there is an $\varepsilon>0$ such that $\beta_k\geq k^{\varepsilon}$
without giving an explicit value of $\varepsilon$ \cite{Ajt08}.
Gama et al. used a completely new method and showed that $\beta_k\geq k/12$ \cite{GamHKN06}.
Since the lower bound on $\beta_k$ shows the limitations of (semi) block $2k$-reduction,
and a better lower bound on it can better quantify the limitations of (semi) block $2k$-reduction,
we will improve the lower bound to around $k/7$ with a simpler method than that used in \cite{GamHKN06}.

Hermite's constant can be used to quantify the length of the shortest nonzero vectors of lattices.
It also has applications in bounding the KZ constant and Schnorr's constant $\beta_k$ from above \cite{Sch87}.
Furthermore, it can be used to derive lower bounds on
the decoding radius of the LLL-aided SIC decoders \cite{WenC18, LuzSL13},
and upper bounds on the orthogonality defect of KZ reduced matrices \cite{LagLS90,WenC18,LyuL17}.
Although Hermite's constant is important, its exact values are known for dimension $1\leq n\leq 8$ and $n=24$ only.
Thus, an upper bound for arbitrary integer $n$ is needed.
In the above applications, Hermite's constant's linear upper bounds play crucial roles.
Hence, in addition to the nonlinear upper bound in \cite{Bli29}, several linear upper bounds on Hermite's constant
have been proposed \cite{LagLS90, NguV10, Neu17}.
Since Hermite's constant has numerous applications,
we will improve the sharpest available linear upper bound provided in \cite{WenC18}.

A mathematical constant $\gamma_{n,\ell}$, where $\ell$ and $n$ are integers and satisfy $1\leq \ell\leq n$,
was introduced by Rankin in 1953 \cite{Ran53} as a generalization of Hermite's constant.
This constant was called Rankin's constant \cite{GamHKN06}.
A reduction procedure called block-Rankin reduction was introduced in \cite{GamHKN06}.
As mentioned in  \cite{GamHKN06}, this reduction may find a shorter lattice vector than semi block $2k$-reduction.
Similar to Schnorr's constant for (semi) block $2k$-reduction,
Rankin's constant is essential to quantify the quality of block-Rankin reduction.
Furthermore, it has been shown in \cite[Theorem 1]{GamHKN06} that $\gamma_{2k,k}^{2/k}\leq \beta_k$.
Therefore, Rankin's and Schnorr's constants have a close relationship and a lower bound on the former
can be used to lower bound the latter.
However, as far as we know, few of Rankin's constants are known.
More exactly, $\gamma_{1,1}=1, \gamma_{4,2}=3/2$ \cite{GamHKN06} and
$\gamma_{6,2}=3^{2/3},\gamma_{8,2}=3,\gamma_{8,3}=\gamma_{8,4}=4$ \cite{SawWO10}.
Therefore, it is essential to bound it, and
Gama et al. \cite{GamHKN06} developed an upper and a lower bound on $\gamma_{2k,k}$ for any integer $k$.
In this paper, we will improve these bounds.

In summary, this paper will improve upper bounds on the KZ and Hermite's constants,
and upper and lower bounds on both Schnorr's constant and Rankin's constant.
More exactly, the main contributions of this paper are summarized as follows,
where the first two contributions were published in the conference paper \cite{WenCW19}.
\begin{itemize}
\item
We improve the existing sharpest linear upper bound on Hermite's constant from $\frac{n}{8}+\frac{6}{5}$ \cite[Theorem 1]{WenC18}
to $\frac{n}{8.5}+2$ (see Theorem \ref{t:gamman})
(the new upper bound is sharper than the former only when $n\geq109$).
The ratio of the new linear upper bound to the nonlinear upper bound developed by Blichfeldt \cite{Bli29} in 1929
is less than 1.0226 for $n\geq 109$ and asymptotically 1.0047 (see Fig.\ \ref{f:HC}).

\item
We improve the existing sharpest upper bound on the KZ constant, that was developed in \cite[Theorem 2]{WenC18}.
Details are given in Theorem \ref{t:KZconstantUB} and Fig.\ \ref{f:KZ}.

\item
We improve the existing sharpest lower bound on Schnorr's constant from $\frac{k}{12}$ \cite{GamHKN06}
to $(\frac{4}{\pi^2\sqrt{k}})^{2/k}\frac{2k}{\pi e^{3/2}}$ (which is around $\frac{k}{7}$ for large $k$)
(see Corollary  \ref{c:SCld}).
We also improve the existing sharpest upper bound on Schnorr's constant
developed in \cite[Theorem 2]{GamHKN06}.
The new one is around 4 times smaller
for $k\geq 40$; see Theorem \ref{t:SC} and Fig.\ \ref{f:SC}  for more details.

\item
We improve the existing sharpest lower and upper bounds on Rankin's constant
 developed in \cite[Theorems 1-3]{GamHKN06}.
The improvements to the lower and upper bounds are around $1.7^{k/2}$ and $2^k$ times;
 see Theorem \ref{t:RC} and Fig.\ \ref{f:RC} for more details.

\end{itemize}

The remainder of the paper is organized as follows.
We develop a new  linear upper bound on Hermite's constant and a new upper bound on the KZ constant
in Sections \ref{s:HC} and \ref{s:KZ}, respectively.
The lower and upper bounds on Schnorr's constant will be improved in \ref{s:sc},
and the lower and upper bounds on Rankin's constant will be improved in \ref{s:rc}.
Finally, we summarize this paper and propose some future works  in Section \ref{s:sum}.

{\it Notation.}
Let $\mathbb{R}^{m\times n}$ and $\mathbb{Z}^{m\times n}$ be the spaces of the $m\times n$ real matrices and integer matrices, respectively.
Boldface lowercase letters denote column vectors and boldface uppercase letters denote matrices.
For a matrix $\boldsymbol{A}$, we use $a_{ij}$ to denote its $(i,j)$ entry
and use $\boldsymbol{A}_{i:j,k:\ell}$ to denote the submatrix containing elements with row indices from $i$ to $j$ and column indices from $k$ to $\ell$.
We use $\Gamma(n)$ and  $\zeta(n)$ to denote the Gamma function and  the Riemann zeta function,
respectively.

\section{A sharper linear bound on Hermite's constant}
\label{s:HC}

This section develops a new linear upper bound, on Hermite's constant,
which is sharper than that presented in \cite[Theorem 1]{WenC18} when $n\geq 109$.

We first introduce the definition of Hermite's constant.
Denote the set of $m\times n$ real matrices with full column rank by $\mathbb{R}_n^{m\times n}$,
then Hermite's constant $\gamma_n$ is defined as
\[
\gamma_n=\sup_{\boldsymbol{A}\in\mathbb{R}_n^{m\times n}}\frac{(\lambda (\boldsymbol{A}))^2}{(\det(\boldsymbol{A}^\top\boldsymbol{A}))^{1/n}},
\]
where $\lambda(\boldsymbol{A})$ represents the length of a shortest nonzero vector of ${\cal L}(\boldsymbol{A})$, i.e.,
$$
\lambda(\boldsymbol{A})=\min_{\boldsymbol{x}\in \mathbb{Z}^{n}\backslash \{\boldsymbol{0}\}} \|\boldsymbol{A}\boldsymbol{x}\|_2.
$$

Although Hermite's constant is a vital constant of lattices, the values of $\gamma_n$
are known only for $n=1,\ldots,8$ \cite{Mar13} and $n=24$ \cite{CohA04} (see also \cite[Table 1]{WenC18}),
and they are summarized in Table \ref{tab:HC}.

\begin{table}[h!]
\caption{Exact values of $\gamma_n$ for $n=1,\ldots,8$ and $n=24$}
\begin{center}
\begin{tabular}{llllllllll}
\toprule
$n$& $1$ & $2$ & $3$& $4$& $5$ & $6$& $7$& $8$& $24$\\
\midrule
$\gamma_n$ & 1 & $\frac{2}{\sqrt{3}} $& $2^{1/3}$ & $\sqrt{2} $& $8^{1/5} $& $(\frac{64}{3})^{1/6}$
&$64^{1/7}$&2 &4\\
\bottomrule
\end{tabular}
\end{center}
\label{tab:HC}
\end{table}

Fortunately, there are some upper bounds on $\gamma_n$ for any $n$ in the literature and the sharpest one is
\begin{equation}
\label{e:blichfeldt}
\gamma_n \leq \frac{2}{\pi} (\Gamma(2+n/2))^{2/n},
\end{equation}
which was given by Blichfeldt \cite{Bli29}.

As explained in Section \ref{s:introduction}, linear upper bounds on $\gamma_n$ are very useful.
There are several linear upper bounds: $\gamma_n \leq \frac{2}{3}n$ (for $n\geq 2$) \cite{LagLS90};
$\gamma_n \leq 1 + \frac{n}{4}$ (for $n\geq 1$) \cite[p.35]{NguV10},
$\gamma_n \leq \frac{n+6}{7} $ (for $n\geq 2$) \cite{Neu17},
and
\begin{equation}
\label{e:gammanWC}
\gamma_n < \frac{n}{8}+\frac{6}{5}, \,\; n\geq 1,
\end{equation}
which was the most recent one given in \cite[Theorem 1]{WenC18}.

The following theorem gives a new linear upper bound on $\gamma_n$,
which is sharper than \eqref{e:gammanWC} when $n\geq109$.

\begin{theorem}
\label{t:gamman}
For $n\geq 1$,
\begin{equation}
\label{e:gamman}
\gamma_n < \frac{n}{8.5}+2.
\end{equation}
\end{theorem}

\begin{proof}
By \eqref{e:blichfeldt}, to show \eqref{e:gamman}, it suffices to show
\[
\left(\Gamma\left(2+\frac{n}{2}\right)\right)^{2/n} <  \frac{\pi(n+17)}{17},
\]
which is equivalent to
\begin{equation}
\label{e:lbd}
\Gamma\left(2+\frac{n}{2}\right) < \left(\frac{\pi(n+17)}{17}\right)^{n/2}.
\end{equation}
Then, to show \eqref{e:gamman}, it is equivalent to show that
\[
\phi(t):=\frac{\left[\frac{\pi}{8.5}(t+8.5)\right]^{t}}{\Gamma\left(2+t\right)} > 1
\]
for $t= 0.5, 1, 1.5, 2, 2.5, \ldots$.

By some direct calculations, one can show that
\[
\phi(t)>1, \mbox{ for }t= 0.5, 1, 1.5, 2, 2.5, \ldots, 310.
\]
Thus, to show \eqref{e:gamman}, we only need to show that $\bar{\phi}(t):=\ln(\phi(t))$
is monotonically  increasing when $t\geq 310$.

By some direct calculations, we have
\begin{align*}
\bar{\phi}'(t)= \ln \left[\frac{\pi}{8.5}(t+8.5)\right] +\frac{t}{t+8.5}-\psi(t+2),
\end{align*}
where $\psi(t+2)$ is the digamma function, i.e., $\psi(t+2)=\Gamma'(t+2)/\Gamma(t+2)$.
Then, to show \eqref{e:gamman}, we only need to show that $\bar{\phi}'(t)\geq 0$ when $t\geq 310$.
To this end, we use the following inequality from \cite[eq. (1.7)]{Bat08}:
\begin{equation}
\label{e:digammaud}
\psi(t+2)\leq \ln (t+e^{1-\gamma}), \mbox{ for } t\geq 0,
\end{equation}
where
\[
\gamma=\lim_{n\rightarrow \infty} (-\ln n + \sum_{k=1}^n 1/k)
\]
is referred to as Euler--Mascheroni constant.
Then, from the expression of $\bar{\phi}'(t)$ given before, we have
 $$
 \bar{\phi}'(t)\geq \rho(t),
$$
 where
\begin{align*}
\rho(t):&=\ln \left[\frac{\pi(t+8.5)}{8.5}\right] +\frac{t}{t+8.5}-\ln (t+e^{1-\gamma})\\
&=\ln (t+8.5)-\frac{8.5}{t+8.5}-\ln (t+e^{1-\gamma}) +\ln\frac{\pi e}{8.5}.
\end{align*}
Since
\begin{align*}
\rho'(t) & = \frac{1}{t+8.5} +\frac{8.5}{(t+8.5)^2} - \frac{1}{t+e^{1-\gamma}} \\
&=  \frac{(t+8.5)(t+e^{1-\gamma})+8.5(t+e^{1-\gamma})}{(t+8.5)^2(t+e^{1-\gamma})} \\
&\quad-  \frac{(t+8.5)^2}{(t+8.5)^2(t+e^{1-\gamma})} \\
&=\frac{e^{1-\gamma}t-(72.25-17e^{1-\gamma})}{(t+8.5)^2(t+e^{1-\gamma})},
\end{align*}
and $\gamma<0.58$ \cite{Leh75}, $\rho'(t)\geq0$ when $t>31$ as
\begin{align*}
&e^{1-\gamma}t-(72.25-17e^{1-\gamma})\\
>&31\times e^{1-\gamma}-(72.25-17e^{1-\gamma})>0.
\end{align*}
Thus, for $t\geq 310$, we have
\[
\bar{\phi}'(t)\geq \rho(t)\geq \rho(310) > 0.0000796>0,
\]
where the third inequality follows form the fact that $\gamma>0.57$ \cite{Leh75}.
\end{proof}

By some simple calculations, one can easily see that the upper bound \eqref{e:gamman} is sharper than
the upper bound \eqref{e:gammanWC} when $n\geq 109$.
When $n\leq 108$,  \eqref{e:gammanWC} is sharper than \eqref{e:gamman}, but their difference is small.
By Stirling's approximation for Gamma function, the right-hand side of \eqref{e:blichfeldt}
is  asymptotically $\frac{n}{\pi e}\approx\frac{n}{8.54} $. Thus, the linear bound given by \eqref{e:gamman}
is very close to the nonlinear upper bound given by \eqref{e:blichfeldt}.
To clearly show the improvement of \eqref{e:gamman} over \eqref{e:gammanWC} and how close \eqref{e:gamman} is to \eqref{e:blichfeldt},
in Fig.\ \ref{f:HC} we plot the ratios of the two upper bounds in \eqref{e:gamman} and  \eqref{e:gammanWC} to Blichfeldt's upper bound in \eqref{e:blichfeldt}, respectively for different rank $n$.

\begin{figure}[!htbp]
\centering
\includegraphics[width=3.4 in]{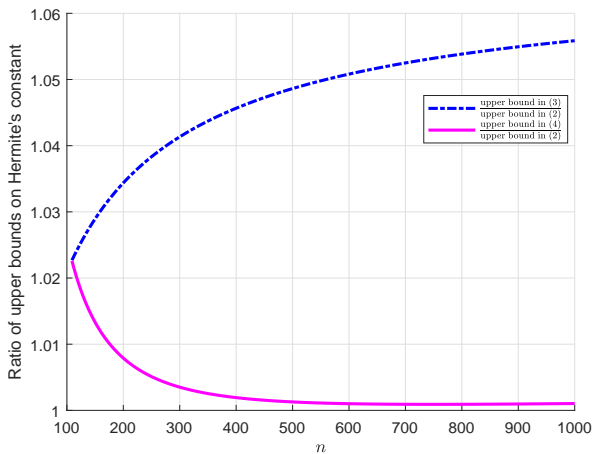}
\caption{
Ratios of two upper bounds to
Blichfeldt's upper bound on Hermite's constant
versus $n$ for $n=109:1:1000$}
\label{f:HC}
\end{figure}

From Fig.\ \ref{f:HC}, one can see that \eqref{e:gamman} is sharper than \eqref{e:gammanWC} for $n\geq 109$,
and the upper bound given by \eqref{e:gamman} is very close to the nonlinear upper bound given by \eqref{e:blichfeldt}.
More exactly, by some calculations, we can show that the ratio of the right-hand side of \eqref{e:gamman}
to the right-hand side of \eqref{e:blichfeldt} is less than 1.0226 for $n\geq 109$ and is asymptotically 1.0047.

In the following, we give some remarks.
\begin{remark}
\label{r:prof}
The main idea of the proof of  \eqref{e:gamman} is different from that of \eqref{e:gammanWC}  \cite{WenC18}.
More exactly, to show \eqref{e:gammanWC}, \cite{WenC18} proved (cf.\ \eqref{e:lbd})
\begin{equation} \label{e:gubdWC}
\Gamma\left(2+\frac{n}{2}\right) < \left(\frac{\pi(n+9.6)}{16}\right)^{n/2}.
\end{equation}
To prove \eqref{e:gubdWC},  \cite{WenC18} first gave an upper bound on $\Gamma\left(2+\frac{n}{2}\right)$
and then showed the right-hand side of
\eqref{e:gubdWC} is larger than this upper bound,
while  the proof of  \eqref{e:lbd} here shows $\ln(\phi(t))$ is a monotonically  increasing function
by using an upper bound on the digamma function (see \eqref{e:digammaud}).
\end{remark}

\begin{remark}
\label{r:LLLdr}
The improved linear upper bound \eqref{e:gamman} on $\gamma_n$ can be used to improve
the lower bound on the decoding radius of the LLL-aided  SIC decoder
that was given in \cite{WenC18} and is an improvement to the one given  in  \cite[Lemma 1]{LuzSL13}.
Since the derivation of the new  lower bound on the decoding radius is straightforward
by following the proof of \cite[Lemma 1]{LuzSL13} and \eqref{e:gamman}, we do not provide the details.
\end{remark}

\begin{remark}
\label{r:LLLOD}
The improved linear upper bound \eqref{e:gamman} on $\gamma_n$ can be used to improve
the upper bound on the orthogonality defect of KZ reduced matrices
that was presented in \cite[Theorem 4]{WenC18}.
Note that the orthogonality defect of a matrix is a good measure of the orthogonality
of the matrix and hence it is often used in characterizing the quality of
an LLL or KZ reduced matrix.
\end{remark}

\section{A sharper bound on the KZ constant}
\label{s:KZ}

In this section, we develop an upper bound, on the KZ constant,
which is sharper than that given by \cite[Theorem 2]{WenC18}.

We first briefly introduce the definition of KZ reduction.
Suppose that $\boldsymbol{A}$ in \eqref{e:latticeA} has the following QR factorization (see, e.g., \cite[Chap.\ 5]{GolV13}):
\begin{equation}
\label{e:QR}
\boldsymbol{A}= \boldsymbol{Q} \boldsymbol{R},
\end{equation}
where $\boldsymbol{Q}\in \mathbb{R}^{m\times n}$ has orthonormal columns (i.e., $\boldsymbol{Q}^T\boldsymbol{Q}=\boldsymbol{I}_n$)
and $\boldsymbol{R}\in \mathbb{R}^{n\times n}$ is a nonsingular upper triangular matrix,
and they are respectively referred to as $\boldsymbol{A}$'s Q-factor and R-factor.
If $\boldsymbol{R}$ in \eqref{e:QR} satisfies:
\begin{align}
\label{e:criteria1}
&{\lvert}r_{ij}{\rvert}\leq\frac{1}{2}{\lvert}r_{ii}{\rvert}, \ \ 1\leq i \leq j-1 \leq n-1,  \\
&{\lvert}r_{ii}{\rvert} =\min_{\boldsymbol{x}\,\in\,\mathbb{Z}^{n-i+1}\backslash \{\boldsymbol{0}\}}\|\boldsymbol{R}_{i:n,i:n}\boldsymbol{x}\|_2, \ \
1\leq i \leq n,
\label{e:criteria2}
\end{align}
then $\boldsymbol{A}$ and $\boldsymbol{R}$ are said to be  KZ reduced.
Given $\boldsymbol{A}\in \mathbb{R}_n^{m\times n}$, which is the set of all $m\times n$ full column rank real matrices, KZ reduction is the process of finding a unimodular matrix
$\boldsymbol{Z}\in \mathbb{Z}^{n\times n}$ such that $\boldsymbol{A}\boldsymbol{Z}$ is  KZ reduced.

Let $\mathcal{S}_{KZ}$ denote the set of all $m\times n$ KZ reduced matrices with full column rank,
then the KZ constant  is  defined as \cite{Sch87}
\begin{equation}
\label{e:KZconstant}
\alpha_n=\sup_{\boldsymbol{A}\in \mathcal{S}_{KZ}} \frac{(\lambda (\boldsymbol{A}))^2}{r_{nn}^2},
\end{equation}
where $\lambda (\boldsymbol{A})$ denotes the length of a shortest nonzero vector of $\mathcal{L}(\boldsymbol{A})$,
and $r_{nn}$ is the last diagonal entry of the R-factor $\boldsymbol{R}$ of $\boldsymbol{A}$ (see \eqref{e:QR}).

As explained in Section \ref{s:introduction}, the KZ constant is an important quantity  for characterizing
some properties of KZ reduced matrices.  However, its exact value is unknown.
Hence, it is useful to find a good upper bound on it.
Schnorr in \cite[Corollary 2.5]{Sch87} proved that
\[
\alpha_n\leq n^{1+\ln n}, \ \ \mbox{for } n\geq 1.
\]
Hanrot and Stehl{\'e} in \cite[Theorem 4]{HanS08} showed that
\[
\alpha_n\leq n\prod_{k=2}^nk^{1/(k-1)} \leq n^{\frac{\ln n}{2}+\mathcal{O}(1)}, \ \ \mbox{for } n\geq 2.
\]
Based on the exact value  of $\gamma_n$  for $1\leq n\leq 8$ and the upper bound
on $\gamma_n$  in \eqref{e:gammanWC} for $n\geq 9$, Wen and Chang in \cite[Theorem 2]{WenC18}
 showed that
\begin{equation}
\label{e:KZconstantUB1}
\alpha_n\leq 7 \left( \frac{1}{8}n +\frac{6}{5}\right)  \left(\frac{n-1}{8}\right)^{\frac{1}{2}\ln((n-1)/8)}, \ \ \mbox{for } n\geq 9.
\end{equation}

In the following theorem we provide a new upper bound on $\alpha_n$ for $n\geq109$,
which is sharper than that in \eqref{e:KZconstantUB1} for $n\geq 111$.
The new bound on $\alpha_n$ is based on the new upper bound on
Hermite's constant $\gamma_n$ in \eqref{e:gamman},
which is sharper than that in \eqref{e:gammanWC}  for $n\geq 109$.

\begin{theorem}
\label{t:KZconstantUB}
The KZ constant $\alpha_n$ satisfies
\begin{equation}
\label{e:KZconstantUB2}
\alpha_n\leq  8.1 \left( \frac{n}{8.5} +2\right)  \left(\frac{2n-1}{17}\right)^{\frac{1}{2}\ln((2n-1)/17)}, \ \ \mbox{for } n\geq 109.
\end{equation}
\end{theorem}

To prove Theorem \ref{t:KZconstantUB}, we need to introduce three lemmas.
The first two lemmas, whose proofs are respectively given in Appendixes \ref{ss:Pint} and \ref{ss:Pf},
are as follows:
\begin{lemma}\label{l:integralbd2}
Suppose that $f(t)$ satisfies $f''(t)\geq 0$ for $t\in [a,b]$.
Then
\begin{equation}
\label{e:integralbd2}
(b-a) f\left(\frac{a+b}{2}\right)\leq \int_a^b f(s) d s.
\end{equation}
\end{lemma}

\begin{lemma}
\label{l:f}
Let
\begin{equation}
\label{e:f}
f(t):=\frac{1}{t} \ln \left(\frac{ t+18}{8.5} \right),\,\; t>0.
\end{equation}
Then $f''(t)>0$.
\end{lemma}

The third one is from \cite[Lemma 2]{WenC18}.
\begin{lemma}\label{l:integralbd}
For $a>b>0$ and $c>0$
\begin{equation}
\label{e:integralbd}
 \int_a^b \frac{\ln( 1+c/t)}{t} d t \leq
 \frac{9}{8} \ln \frac{b(3a+2c)}{a(3b+2c)} + \frac{c(b-a)}{4ab}.
\end{equation}
\end{lemma}

In the following, we prove Theorem \ref{t:KZconstantUB}
by following the proof of \cite[Theorem 2]{WenC18}.
\begin{proof}
According to the proof of \cite[Cor. 2.5]{Sch87},
\begin{equation}
\label{e:KZUB11}
\alpha_n\leq \gamma_n\prod_{k=2}^n\gamma_k^{1/(k-1)}.
\end{equation}
By \cite[(53)]{WenC18}, we have
\begin{equation}
\label{e:KZUB12}
\prod_{k=2}^8\gamma_k^{1/(k-1)}= 2^{\frac{827}{420}} 3^{-\frac{8}{15}}.
\end{equation}
By \eqref{e:gammanWC}, we obtain
\begin{equation}
\label{e:KZUB13}
\prod_{k=9}^{108}\gamma_k^{1/(k-1)}\leq \prod_{k=9}^{108}\left(\frac{k}{8}+\frac{6}{5}\right)^{1/(k-1)}< 79.06.
\end{equation}

In the following, we use Theorem \ref{t:gamman} to bound $\prod_{k=109}^n\gamma_k^{1/(k-1)}$ from above.
By Theorem \ref{t:gamman}, we obtain
\begin{align}
  & \prod_{k=109}^n\gamma_k^{1/(k-1)}  \notag \\
\leq &  \prod_{k=109}^n\left(\frac{k}{8.5} + 2\right)^{1/(k-1)}
=    \prod_{k=108}^{n-1}\left(\frac{ k+18}{8.5} \right)^{1/k}    \nonumber  \\
= & \exp\left[\sum_{k=108}^{n-1}\frac{1}{k}\ln \left(\frac{ k+18}{8.5} \right)\right]\nonumber\\
\overset{(a)}{\leq}&\exp\left(\sum_{k=108}^{n-1}\int_{k-0.5}^{k+0.5}\frac{1}{t} \ln  \left(\frac{ t+18}{8.5} \right)dt\right)\label{e:proofa} \\
=&\exp\!\left(\int_{107.5}^{n-0.5}\frac{1}{t} \ln\left(\frac{t+18}{t}\frac{t}{8.5}\right)dt\right)\nonumber\\
 = &\exp\!\left(\int_{107.5}^{n-0.5} \frac{1}{t} \ln \left(1 +\frac{18}{t}\right) dt \! \right )
   \times \exp\! \left(\int_{107.5}^{n-0.5}\frac{\ln(t/8.5)}{t}dt \! \right),
\label{e:kz2terms}
\end{align}
where (a) follows from Lemma \ref{l:integralbd2} with $a=k-0.5, b=k+0.5$ and Lemma \ref{l:f}.

Now we bound the two factors on the right-hand side of \eqref{e:kz2terms} from above.
By Lemma \ref{l:integralbd}, we obtain
\begin{align}
& \exp\left( \int_{107.5}^{n-0.5}\frac{1}{t} \ln  \left(1 +\frac{18}{t}\right) dt \right)  \nonumber  \\
\leq  &  \exp\left( \frac{9}{8} \ln \frac{358.5(n-0.5)}{107.5(3(n-0.5)+36)} +\frac{18(n-108)}{430(n-0.5)}  \right)  \nonumber  \\
\leq  &  \exp\left( \frac{9}{8} \ln \frac{358.5(n-0.5)}{107.5\times3(n-0.5)} +\frac{18(n-108)}{430(n-108)}  \right)  \nonumber  \\
=   &   \left(\frac{119.5}{107.5}\right)^{9/8} \exp\left( \frac{9}{215}\right). \label{e:1stub}
\end{align}

By a direct calculation, we have
\begin{align}
 & \exp\left(\int_{107.5}^{n-0.5}\frac{\ln(t/8.5)}{t}dt\right) \nonumber\\
=\ & \exp\left(\frac{\ln^2((n-0.5)/8.5)}{2}-\frac{\ln^2(107.5/8.5)}{2}\right)   \nonumber\\
 =\ & \left(\frac{n-0.5}{8.5}\right)^{\frac{1}{2}\ln((n-0.5)/8.5)}
 \left(\frac{8.5}{107.5}\right)^{\frac{1}{2}\ln(107.5/8.5)}. \label{e:2ndub}
\end{align}

Then combining  \eqref{e:KZUB11}-\eqref{e:2ndub} and \eqref{e:gamman}, we obtain that
for $n\geq 109$,
\begin{align*}
 \alpha_n
\leq \ & 79.06\times 2^{\frac{827}{420}} 3^{-\frac{8}{15}}
       \left(\frac{119.5}{107.5}\right)^{9/8} \exp\left( \frac{9}{215}\right) \\ &  \times       \left(\frac{8.5}{107.5}\right)^{\frac{1}{2}\ln\frac{107.5}{8.5}}
\left( \frac{n}{8.5} +2\right)  \left(\frac{n-0.5}{8.5}\right)^{\frac{1}{2}\ln(\frac{n-0.5}{8.5})} \\
< \ & (8.0911 \cdots)   \left( \frac{n}{8.5} +2\right)  \left(\frac{n-0.5}{8.5}\right)^{\frac{1}{2}\ln(\frac{n-0.5}{8.5})} \\
< \ & 8.1 \left( \frac{n}{8.5} +2\right)  \left(\frac{2n-1}{17}\right)^{\frac{1}{2}\ln((2n-1)/17)}.
\end{align*}
\end{proof}

\begin{remark}
Note that although the method used in the proof of Theorem \ref{t:KZconstantUB} is similar to
that used in the proof of
\cite[Theorem 2]{WenC18}, there are some differences between them.
The main difference is here Lemma \ref{l:integralbd2} has been used in deriving the inequality \eqref{e:proofa},
while the decreasing property of the integrand was used to derive the corresponding
inequality in the proof of \cite[Theorem 2]{WenC18}.
\end{remark}

To clearly see the improvement of \eqref{e:KZconstantUB2} over \eqref{e:KZconstantUB1},
we draw the ratio of the upper bound in \eqref{e:KZconstantUB2}
to the upper bound  in  \eqref{e:KZconstantUB1}
for different rank $n$ in Fig.\ \ref{f:KZ}.
The figure shows that \eqref{e:KZconstantUB2} significantly outperforms  \eqref{e:KZconstantUB1},
and the improvement becomes more significant as $n$ gets larger.

\begin{figure}[!htbp]
\centering
\includegraphics[width=3.4 in]{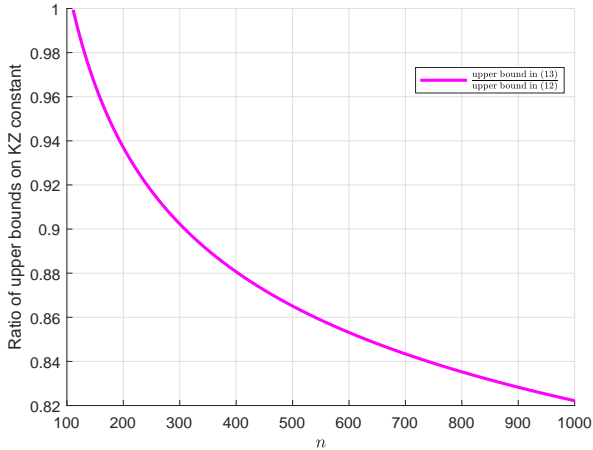}
\caption{Ratio of the upper bound in \eqref{e:KZconstantUB2}
to the upper bound in \eqref{e:KZconstantUB1} on the KZ constant
versus $n$ for $n=111:1:1000$}
\captionsetup{justification=centering}
\label{f:KZ}
\end{figure}

In the following we give remarks about some applications of Theorem \ref{t:KZconstantUB}.

\begin{remark}
\label{r:KZpf}
As \cite[Remark 2]{WenC18}, we can use the improved upper bound \eqref{e:KZconstantUB2} on $\alpha_n$
to derive new upper bounds on the proximity factors of the KZ-aided SIC decoder,
which will be  sharper than those given in \cite[Remark 2]{WenC18}.
Since the derivations are straightforward, we omit them.
\end{remark}

\begin{remark}
\label{r:KZdr}
We can use  \eqref{e:KZconstantUB2} and follow the proof of \cite[Lemma 1]{LuzSL13} to derive a lower bound on  the decoding radius of the KZ-aided SIC decoder, which
will be tighter than that given in \cite[Remark 3]{WenC18} when $n\geq 111$.
\end{remark}

\begin{remark}
\label{r:KZod}
By following the proof of \cite[Theorem 3]{WenC18} and using  \eqref{e:KZconstantUB2},
we can also develop new upper bounds on the lengths of the KZ reduced matrices,
which will be tighter than those given in \cite[Theorem 3]{WenC18} when $n\geq 111$.
\end{remark}

\section{A sharper bound on Schnorr's constant}
\label{s:sc}

In this section, we develop upper and lower bounds, on Schnorr's constant, that are sharper than those presented in \cite{GamHKN06}.

The block $2k$-reduction and semi block $2k$-reduction were introduced by Schnorr \cite{Sch87}.
A matrix $\boldsymbol{A}\in \mathbb{R}^{m\times (\ell k)}$ is called block $2k$-reduced if the
R-factor $\boldsymbol{R}$ (see \eqref{e:QR}) of the QR factorization of $\boldsymbol{A}$ satisfies \eqref{e:criteria1} and all
\[
\boldsymbol{R}_{ik+1:(i+2)k,ik+1:(i+2)k},\,\; \mbox{ for } i=0,\ldots,\ell-2
\]
are KZ reduced.

Schnorr also introduced a constant $\beta_k$ to globally measure the drop of the diagonal entries of the R-factor $\boldsymbol{R}$
of the QR factorization of block $2k$-reduced matrices. Mathematically,
\begin{equation*}
\beta_k=\sup_{\boldsymbol{A}\in \mathcal{S}_{\mathrm{KZ}}} \left(\frac{\prod_{i=1}^kr_{ii}}{\prod_{i=k+1}^{2k}r_{ii}}\right)^{2/k},
\end{equation*}
where $\mathcal{S}_{\mathrm{KZ}}$ denotes the set of all $m\times 2k$ KZ reduced matrices with full column rank,
and $r_{ii}$ for $1\leq i\leq n$ are the diagonal entries of the R-factor $\boldsymbol{R}$ of $\boldsymbol{A}$ (see \eqref{e:QR}).

A matrix $\boldsymbol{A}\in \mathbb{R}^{m\times (\ell k)}$ is called semi block $2k$-reduced if the
R-factor $\boldsymbol{R}$ (see \eqref{e:QR}) of the QR factorization of $\boldsymbol{A}$ satisfies
\begin{align*}
\prod_{j=1}^kr_{ik+j,ik+j}^2&\leq \frac{4}{3}\beta_k^k\prod_{j=1}^kr_{(i+1)k+j,(i+1)k+j}^2,\\
r_{ik,ik}^2&\leq 2\,r_{ik+1,ik+1}^2,
\end{align*}
for $ i=0,\ldots,\ell-1$ and $\boldsymbol{R}_{ik+1:(i+1)k,ik+1:(i+1)k}$ are KZ reduced for $ i=0,\ldots,\ell-1$.

The semi block $2k$-reduction is a relaxed reduction of block $2k$-reduction,
and it can be performed in polynomial time when reducing an integer matrix with suitable $k$.
It is a generalization of and more powerful than LLL reduction.

The constant $\beta_k$ was called Schnorr's constant \cite{GamHKN06} for example.
It is essential to characterize the output quality of block $2k$-reduction
and to define semi block $2k$-reduction.
However, there is no formula to calculate it.
It is important to bound it, and a better upper bound on it can better quantify the quality of (semi) block $2k$-reduction.
It has been shown in \cite[Theorem 2.7]{Sch87} that
\[
\beta_k\leq 4k^2, \ \ \mbox{for } k\geq 1,
\]
and in particular, $\beta_1=\frac{4}{3}$, $\beta_2\leq1.59$, $\beta_3\leq1.91$ and $\beta_4\leq2.25$.
This upper bound was improved to
\begin{equation}
\label{e:SCold}
\beta_k\leq \left(1+\frac{k}{2}\right)^{2 \ln2+1/k}
\end{equation}
in  \cite[Theorem 2]{GamHKN06}.
Theorem \ref{t:SC} below provides a new upper bound on $\beta_k$.
The new bound is around $1/4$ of the one given in \eqref{e:SCold} for $k\geq 40$.

On the other hand, a lower bound on $\beta_k$ is given by Gama et al. \cite{GamHKN06}:
\begin{equation} \label{e:betalb}
\beta_k\geq k/12.
\end{equation}
This lower bound  shows the limitations of (semi) block $2k$-reduction
and a sharper one can better quantify the limitation.
We will improve the lower bound to $\left(\frac{4}{\pi^2\sqrt{k}}\right)^{2/k}\frac{2k}{\pi e^{3/2}}$
which is around $k/7$ for large $k$.
Since the lower bound is derived from the lower bound on Rankin's constant,
which will be introduced in the next section, we will develop the lower bound on $\beta_k$
in Corollary \ref{c:SCld} in the next section.

Based on the exact values of $\gamma_k$ (see Table \ref{t:gamman}),
Schnorr showed that $\beta_1=4/3, \beta_2\leq 1.59, \beta_3\leq 1.91, \beta_4\leq 2.25$ \cite{Sch87}.
Hence, we derive a sharper upper bound than \eqref{e:SCold} for $k\geq 5$ in the following theorem.

\begin{theorem}
\label{t:SC}
For $k\geq 5$, Schnorr's constant $\beta_k$ satisfies
\begin{equation}
\label{e:SC}
\beta_k\leq \frac{2^{3\ln2}}{17^{2\ln2}}\left( \frac{4k-1}{17}\right)^{\frac{1}{2k-1}}e^{\frac{18}{k}}(k-0.5)^{2\ln2}.
\end{equation}
\end{theorem}

\begin{proof}
By \cite[(2.6)]{Sch87}, we have
\[
\beta_k\leq \prod_{i=1}^{k}\gamma_{2k-i+1}^{2/(2k-i)}=\prod_{i=k}^{2k-1}\gamma_{i+1}^{2/i}.
\]
Then, by \eqref{e:gamman}, for $k\geq 5$, we obtain
\begin{align}
\beta_k\leq &  \prod_{i=k}^{2k-1}\left(\frac{ i+18}{8.5} \right)^{2/i}
=  \exp\left[2\sum_{i=k}^{2k-1}\frac{1}{i}\ln \left(\frac{ i+18}{8.5} \right)\right]\nonumber\\
\overset{(a)}{\leq}&\exp\left(2\sum_{i=k}^{2k-1}\int_{i-0.5}^{i+0.5}\frac{1}{t} \ln  \left(\frac{ t+18}{8.5} \right)dt\right)\nonumber\\
=&\exp\!\left(2\int_{k-0.5}^{2k-0.5}\frac{1}{t} \ln\left(\frac{t+18}{t}\frac{t}{8.5}\right)dt\right)\nonumber\\
 = &\exp\!\left(2\int_{k-0.5}^{2k-0.5} \frac{1}{t} \ln \left(1 +\frac{18}{t}\right) dt \! \right ) \label{e:sc2terms} \\
 &\times \exp\! \left(2\int_{k-0.5}^{2k-0.5}\frac{\ln(t/8.5)}{t}dt \! \right), \nonumber
\end{align}
where (a) follows from Lemma \ref{l:integralbd2}  with $a=i-0.5, b=i+0.5$ and Lemma \ref{l:f}.

Now we bound the two factors on the right-hand side of the equality \eqref{e:sc2terms}.
By Lemma \ref{l:integralbd}, we obtain
\begin{align}
& \exp\left(2 \int_{k-0.5}^{2k-0.5}\frac{1}{t} \ln  \left(1 +\frac{18}{t}\right) dt \right)  \nonumber  \\
\leq  &  \exp\left( \frac{9}{4} \ln \frac{(2k-0.5)(k+11.5)}{(k-0.5)(2k+11.5)} +\frac{9k}{(k-0.5)(2k-0.5)}  \right)  \nonumber  \\
<   &  \exp\left( \frac{9}{4} \ln \left(1+\frac{12k}{2k^2+10.5k-5.75}\right) +\frac{18k}{4k^2-3k}  \right)  \nonumber  \\
\leq \ &  \exp\left( 9 \left(\frac{6k}{4k^2+21k-11.5}+\frac{2}{4k-3} \right) \right)  \nonumber  \\
=  &  \exp\left( 9 \times\frac{32k^2+24k-23}{16k^3+72k^2-109k+34.5} \right)  \nonumber  \\
\leq   &   \exp\left(\frac{18}{k}\right),
\label{e:sc1stub}
\end{align}
where the last inequality is because $k\geq 5$.

By some direct calculations, we have
\begin{align}
 & \exp\left(2\int_{k-0.5}^{2k-0.5}\frac{\ln(t/8.5)}{t}dt\right) \nonumber\\
=\ &\exp\left(2\int_{k-0.5}^{2k-1}\frac{\ln(t/8.5)}{t}dt\right)
\exp\left(2\int_{2k-1}^{2k-0.5}\frac{\ln(t/8.5)}{t}dt\right) \nonumber\\
\leq\ & \exp\left(\ln^2((2k-1)/8.5)-\ln^2((k-0.5)/8.5)\right)   \nonumber\\
\ &\times\exp\left(2\int_{2k-1}^{2k-0.5}\frac{\ln ((2k-0.5)/8.5)}{2k-1}dt\right) \nonumber\\
= & \exp\left(\ln((2k-1)(k-0.5)/8.5^2)\times\ln2\right)   \nonumber\\
\ &\times\exp\left(\frac{\ln ((4k-1)/17)}{2k-1}\right) \nonumber\\
=\ & \left[\frac{(2k-1)(k-0.5)}{8.5^2}\right]^{\ln2}\times\left( \frac{4k-1}{17}\right)^{1/(2k-1)} \nonumber\\
<\ &\frac{2^{3\ln2}}{17^{2\ln2}}\times\left( \frac{4k-1}{17}\right)^{1/(2k-1)}(k-0.5)^{2\ln2}.
 \label{e:sc2ndub}
\end{align}
Then by  \eqref{e:sc2terms}-\eqref{e:sc2ndub}, we obtain \eqref{e:SC}.

\end{proof}

From the above proof we can see that \eqref{e:SC} is obtained based on the new upper bound \eqref{e:gamman}
on Hermite's constant $\gamma_n$, while the bound \eqref{e:SCold} is based on the upper bound
$\gamma_n \leq 1 + \frac{n}{4}$ \cite[p.35]{NguV10}.
It is worthwhile mentioning that the method for deriving \eqref{e:SC} is different from
that for \eqref{e:SCold}, and it is impossible to get an upper bound as sharp as \eqref{e:SC}
if we use the same method for deriving \eqref{e:SCold} even if we use \eqref{e:gamman}.
Indeed, it is shown in \cite[Theorem 2]{GamHKN06} that the asymptotic upper bound is
$\beta_k \leq \frac{1}{10}k^{2\ln2}$ based on the asymptotic upper bound $\gamma_n \leq\frac{1.744n}{2\pi e}$,
which is sharper than \eqref{e:gamman},
while it is not hard to see that the upper bound given in \eqref{e:SC} tends to $0.0833k^{2\ln 2}$.

The upper bound on $\beta_k$ given by \eqref{e:SC} is complicated,
so in the following, we give a less tight, but simpler upper bound on $\beta_k$:

\begin{corollary}
\label{c:SC}
For $k\geq 5$, Schnorr's constant $\beta_k$ satisfies
\begin{equation}
\label{e:SC2}
\beta_k\leq 0.08698e^{\frac{18}{k}}(k-0.5)^{2\ln2}.
\end{equation}
\end{corollary}

\begin{proof}
See Appendix \ref{ss:Pc}.
\end{proof}

To clearly see the improvement of \eqref{e:SC} over \eqref{e:SCold} and how close  \eqref{e:SC2} is to \eqref{e:SC},
we draw the ratios of the upper bounds in \eqref{e:SC} and \eqref{e:SC2} to that in \eqref{e:SCold}
for $k=5:5:1000$ in Fig.\ \ref{f:SC}.
From the figure, we can see that \eqref{e:SC} significantly outperforms  \eqref{e:SCold},
and the improvement is around 4 times when $k\geq 40$.
Fig.\ \ref{f:SC} also shows that the upper bound given by \eqref{e:SC2} is very close to that given by \eqref{e:SC}.

\begin{figure}[!htbp]
\centering
\includegraphics[width=3.4 in]{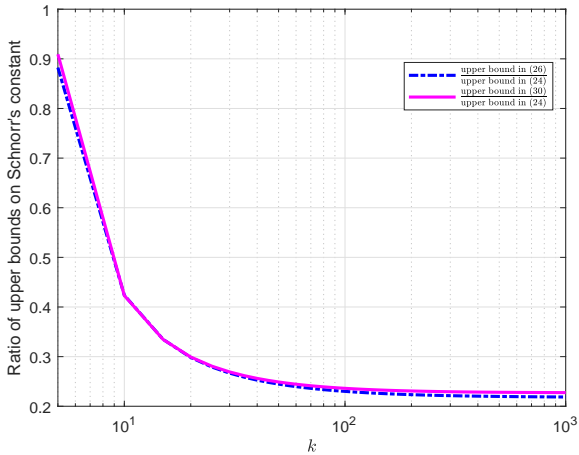}
\caption{
Ratios of two upper bounds in  \eqref{e:SC} and \eqref{e:SC2}
to that in \eqref{e:SCold} on Schnorr's constant
versus block size $k$ for $k=5:5:1000$}
\captionsetup{justification=centering}
\label{f:SC}
\end{figure}

\section{Sharper bounds on Rankin's constant}
\label{s:rc}

In this section, we present a new upper bound and a new lower bound on Rankin's constant.
The new bounds are  exponentially sharper than the existing sharpest bounds
that were developed in \cite{GamHKN06}.

Rankin's constant was defined by Rankin in 1953 \cite{Ran53} as a generalization of Hermite's constant.
For any integer $1\leq \ell\leq n$,  Rankin's constant $\gamma_{n,\ell}$ is defined as \cite{Ran53}:
\begin{equation*}
\gamma_{n,\ell}=\max_{\boldsymbol{A} \in \mathbb{R}_n^{m\times n}}\min_{\boldsymbol{B}\in \mathcal{S}_{\boldsymbol{A},\ell}}
\frac{\det(\boldsymbol{B}^\top\boldsymbol{B})}{\det(\boldsymbol{A}^\top\boldsymbol{A})^{\ell/n}},
\end{equation*}
where $\mathbb{R}_n^{m\times n}$ denotes the set of all $m\times n$ real matrices with rank $n$,
and  $\mathcal{S}_{\boldsymbol{A},\ell}$ denotes the set of all $m\times \ell$ full column rank matrices whose columns belong to $\mathcal{L}(\boldsymbol{A})$.

The quantity $\gamma_{n,\ell}$ is a measure of the output quality of block-Rankin reduction, which was introduced in \cite{GamHKN06}.
However, as far as we know, only $\gamma_{1,1}, \gamma_{4,2}$ \cite{GamHKN06} and
$\gamma_{6,2},\gamma_{8,2},\gamma_{8,3}, \gamma_{8,4}$  are known \cite{SawWO10}.
Therefore, it is useful to bound it. It has been shown in \cite[Theorems 1-3]{GamHKN06} that:
\begin{equation}
\label{e:RCold}
\left(\frac{k}{12}\right)^{k/2}\leq\gamma_{2k,k}\leq \left(1+\frac{k}{2}\right)^{k \ln2+1/2}, \,\;k\geq 2.
\end{equation}

As mentioned in \cite{GamHKN06}, the lower bound on $\gamma_{2k,k}$ suggests
that the approximation of the length of the shortest nonzero vector found
by any block-reduction algorithm (including
Schnorr's semi block $2k$-reduction) to that of the shortest nonzero vector of the lattice
based on the LLL strategy is limited.
In the following, we develop a sharper upper bound and a sharper lower bound on $\gamma_{2k,k}$.
Before presenting the new bounds, we need to introduce the following two lemmas.

\begin{lemma}
\label{l:rie}
For any integer $k\geq 2$, the Riemann zeta function $\zeta(s)=\sum_{n=1}^{\infty}\frac{1}{n^s}$ satisfies
\begin{equation}
\label{e:rie}
\frac{\prod_{i=k+1}^{2k}\zeta(i)}{\prod_{i=2}^{k}\zeta(i)}>\pi^{4-2k}2^{3k-7}.
\end{equation}
\end{lemma}

\begin{proof}
See Appendix \ref{ss:Prie}.
\end{proof}

\begin{lemma}
\label{l:gamma}
For any integer $k\geq 2$, we have
\begin{equation}
\label{e:gamma}
\frac{\prod_{i=k+1}^{2k}\Gamma(\frac{i}{2}+1)}{\prod_{i=2}^{k}\Gamma(\frac{i}{2}+1)}
\geq\sqrt{\pi}2^{\frac{k^2}{2}}e^{-\frac{3k^2+4}{4}}k^{\frac{2k^2-2k-1}{4}}\sqrt{\frac{(2k)!}{k!(k+1)!}}.
\end{equation}
\end{lemma}

\begin{proof}
See Appendix \ref{ss:Pgamma}.
\end{proof}

As $\gamma_{4,2}$ \cite{GamHKN06} and $\gamma_{8,4}$  are known \cite{SawWO10},
and $\frac{4}{\sqrt{3}}\leq \gamma_{6,3}\leq \sqrt{6}$ \cite{SawWO10},
we provide a sharper upper bound and a sharper lower bound on $\gamma_{2k,k}$
than \eqref{e:RCold} for $k\geq 5$ in the following theorem.

\begin{theorem}
\label{t:RC}
For any integer $k\geq 5$, Rankin's constant $\gamma_{2k,k}$ satisfies
\begin{align}
\label{e:RC}
\gamma_{2k,k}&\geq \frac{4}{\pi^2\sqrt{k}}\left(\frac{2k}{\pi e^{3/2}}\right)^{k/2},\\
\gamma_{2k,k}& \leq e^{9}(0.0833)^{k/2}\left( \frac{4k-1}{17}\right)^{\frac{k}{4k-2}}(k-0.5)^{k\ln2}.
\label{e:RCUB}
\end{align}
\end{theorem}

\begin{proof}
 By \cite[Theorem 1]{GamHKN06},
\begin{equation} \label{e:gammabeta}
\gamma_{2k,k}\leq \beta_k^{k/2},
\end{equation}
then \eqref{e:RCUB} follows from Theorem \ref{t:SC} and \eqref{e:con1}.

In the following, we prove \eqref{e:RC}. By \cite[Corollary 1]{Thu98}, for $k\geq 2$,
\begin{align}
\label{e:RC1}
\gamma_{2k,k}^k\geq & \ 4k\frac{\prod_{i=k+1}^{2k}\zeta(i)\Gamma(\frac{i}{2}+1)/(i\pi^{i/2})}
{\prod_{i=2}^{k}\zeta(i)\Gamma(\frac{i}{2}+1)/(i\pi^{i/2})}\nonumber\\
=& \ 4k\frac{(k!)^2}{(2k)!}\pi^{-\frac{k^2+1}{2}}
\frac{\prod_{i=k+1}^{2k}\zeta(i)}{\prod_{i=2}^{k}\zeta(i)}
\frac{\prod_{i=k+1}^{2k}\Gamma(\frac{i}{2}+1)}{\prod_{i=2}^{k}\Gamma(\frac{i}{2}+1)}\nonumber\\
\geq&\ \frac{\pi^{4}}{2^5e}
2^{\frac{k^2}{2}+3k}e^{-\frac{3k^2}{4}}\pi^{-\frac{k^2}{2}-2k}
k^{\frac{2k^2-2k-1}{4}}\frac{k!}{\sqrt{(2k)!}}\nonumber\\
\geq&\ 2^{\frac{k^2}{2}+3k}e^{-\frac{3k^2}{4}}\pi^{-\frac{k^2}{2}-2k}
k^{\frac{2k^2-2k-1}{4}}\frac{k!}{\sqrt{(2k)!}},
\end{align}
where the second inequality follows from Lemmas \ref{l:rie} and \ref{l:gamma},
and $k\geq \sqrt{k+1}$ for $k\geq 2$.

By Stirling's approximation, we have
\[
\sqrt{2\pi}k^{k+\frac{1}{2}}e^{-k}\leq k!\leq ek^{k+\frac{1}{2}}e^{-k}.
\]
Hence,
\[
\frac{k!}{\sqrt{(2k)!}}\geq\frac{\sqrt{2\pi}k^{k+\frac{1}{2}}e^{-k}}{\sqrt{e}(2k)^{k+\frac{1}{4}}e^{-k}}
=2^{-k+\frac{1}{4}}\pi^{\frac{1}{2}}e^{-\frac{1}{2}}k^{\frac{1}{4}}
>2^{-k}k^{\frac{1}{4}}.
\]
Then, by \eqref{e:RC1}, we have
\begin{align*}
\gamma_{2k,k}^k\geq&2^{\frac{k^2}{2}+2k}e^{-\frac{3k^2}{4}}\pi^{-\frac{k^2}{2}-2k}k^{\frac{2k^2-2k}{4}}\\
 = &\left(\frac{2k}{\pi e^{3/2}}\right)^{k^2/2}\left(\frac{4}{\pi^2\sqrt{k}}\right)^{k},
\end{align*}
leading to \eqref{e:RC}.
\end{proof}

Here we make some remarks.
Gama et al.'s upper bound on $\gamma_{2k,k}$ in \eqref{e:RCold}  was obtained by combining their
upper bound on $\beta_k$ in  \eqref{e:SCold} with the inequality \eqref{e:gammabeta},
while our new bound on $\gamma_{2k,k}$ in \eqref{e:RCUB} is obtained by combining our
upper bound on $\beta_k$ in \eqref{e:SC} with the same inequality.
Since it was shown in the last section that
the upper bound on $\beta_k$ in \eqref{e:SC} is around 4 times smaller than that in \eqref{e:SCold},
the upper bound on $\gamma_{2k,k}$ in \eqref{e:RCUB} is around $2^k$ times smaller than
that in \eqref{e:RCold}.

By Corollary \ref{c:SC} and \eqref{e:gammabeta}, we can get the following upper bound, on $\gamma_{2k,k}$,
which is less sharp than \eqref{e:RCUB} but it is simpler,
\[
\gamma_{2k,k} \leq e^{9}(0.08698)^{k/2}(k-0.5)^{k\ln2}.
\]

Our new lower bound on $\gamma_{2k,k}$ in \eqref{e:RC} is sharper than that in \eqref{e:RCold},
because their ratio
\[
\frac{\frac{4}{\pi^2\sqrt{k}}\left(\frac{2k}{\pi e^{3/2}}\right)^{k/2}}{\left(\frac{k}{12}\right)^{k/2}}
=\frac{4}{\pi^2\sqrt{k}}\left(\frac{24}{\pi e^{3/2}}\right)^{k/2}
>\frac{4}{\pi^2\sqrt{k}}\times1.7^{k/2}.
\]
To clearly see the improvement, we draw the ratio of the upper bound in \eqref{e:RC} to that of \eqref{e:RCold} for $k=5:5:50$ in Fig.\ \ref{f:RC},
Fig.\ \ref{f:RC} shows that \eqref{e:RC} exponentially  outperforms  \eqref{e:RCold}, and the improvement becomes more significant as $k$ gets larger.

\begin{figure}[!htbp]
\centering
\includegraphics[width=3.4 in]{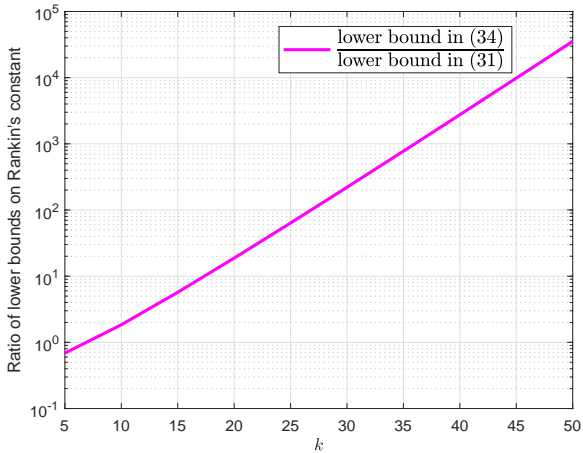}
\caption{
Ratio of the lower bound in \eqref{e:RC} to that in \eqref{e:RCold}
on Rankin's constant
versus  block size $k$  for $k=5:5:50$}
\captionsetup{justification=centering}
\label{f:RC}
\end{figure}

A byproduct of Theorem \ref{t:RC} is a lower bound on $\beta_k$,
given in the following corollary, which can be obtained from \eqref{e:gammabeta} and \eqref{e:RC}.
\begin{corollary}
\label{c:SCld}
For $k\geq 2$, Schnorr's constant $\beta_k$ satisfies
\begin{equation}
\label{e:SCld}
\beta_k\geq \left(\frac{4}{\pi^2\sqrt{k}}\right)^{2/k}\frac{2k}{\pi e^{3/2}}.
\end{equation}
\end{corollary}

By some simply calculations, one can show that the new lower bound on $\beta_k$ given in \eqref{e:SCld}
is sharper than that given in \eqref{e:betalb}
when $k\geq 8$, and the improvement tends to $\frac{2k}{\pi e^{3/2}}/(k/12)> 1.70$ times.

To clearly see the improvement, we draw the ratio of the lower bound in  \eqref{e:SCld} to that of \eqref{e:betalb}
for $k=1:1:500$ in Fig.\ \ref{f:SCLBD}.
From Fig.\ \ref{f:SCLBD}, we can see that the new lower bound given in \eqref{e:SCld} is sharper than that given in \eqref{e:betalb}
when $k$ is not very small and the improvement tends to around 1.7 times.
\begin{figure}[!htbp]
\centering
\includegraphics[width=3.4 in]{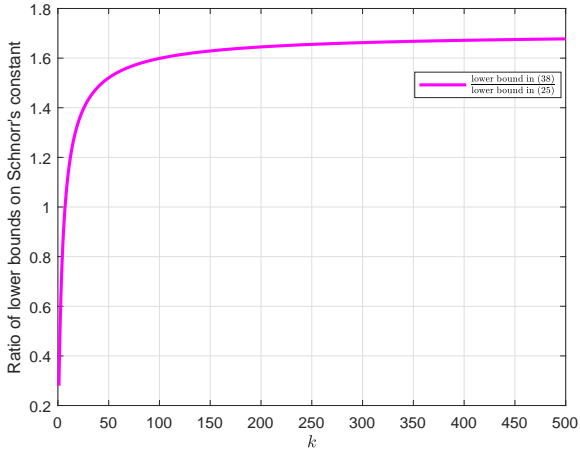}
\caption{
Ratio of the lower bound in \eqref{e:SCld}
to that in \eqref{e:betalb} on Schnorr's constant
versus block size $k$ for $k=1:1:500$}
\captionsetup{justification=centering}
\label{f:SCLBD}
\end{figure}

\section{Summary and Future Work} \label{s:sum}

We have derived upper bounds or lower bounds on
the KZ constant, Hermite's constant, Schnorr's constant and Rankin's constant, which are sharper or much sharper than the corresponding bounds in the literature.
Specifically, we first developed a new linear upper bound on Hermite's constant, and then used it to drive a new upper bound on the KZ constant
and a new upper bound on Schnorr's constant,
which was then used to derive a new upper bound on Rankin's constant.
We also developed a new lower bound on Rankin's constant separately,
and then used it to derive a new lower bound on Schnorr's constant.

The Block KZ (BKZ) reduction, slide reduction \cite{GamN08}, self-dual BKZ reduction \cite{MicW16}
are among the practical lattice reductions used in cryptanalysis.
In the future, we will investigate whether the techniques developed in this paper can be applied or extended to better quantify the output qualities of these reductions.

\section*{Acknowledgment}
We are grateful to the two anonymous referees for their valuable and thoughtful suggestions, which improve the presentation of our work significantly.

\begin{appendices}

  \section{Proof of Lemma~\ref{l:integralbd2}}
  \label{ss:Pint}

  \begin{proof}
  The left hand side of \eqref{e:integralbd2} is referred to as the midpoint rule
  for approximating  the integral on the right hand side in numerical analysis.
  It is well known that
  \begin{equation} \label{e:midpoint}
   \int_a^b f(s) d s - (b-a) f\left(\frac{a+b}{2}\right)
  = \frac{1}{24}(b-a)^3 f''(z)
  \end{equation}
  for some $z \in (a,b)$.
  This formula can be easily proved as follows.
  By Taylor's theorem,
  \begin{align*}
  f(s) = & f\left(\frac{a+b}{2}\right) + f'\left(\frac{a+b}{2}\right) \Big(s- \frac{a+b}{2}\Big) \\
  & + \frac{1}{2} f''(\zeta(s)) \Big(s- \frac{a+b}{2}\Big)^2,
  \end{align*}
  where $\zeta(s)$ depends on $s\in (a, b)$.
  Integrating both sides of the above equality over $[a,b]$
  and using the Mean-Value-Theorem for integrals
  immediately lead to \eqref{e:midpoint}.
  Then using the condition that $f''(t)\geq 0$ for $t\in [a,b]$,
  we obtain \eqref{e:integralbd2}.
  \end{proof}

  \section{Proof of Lemma~\ref{l:f}}
  \label{ss:Pf}

  \begin{proof}
  By some direct calculations, we have
  \begin{align*}
  f'(t)&\!=\frac{1}{t}\left[\frac{1}{t+18}-\frac{1}{t} \ln \left(\frac{ t+18}{8.5} \right)\right].
  \end{align*}
  Then,
  \[
  f''(t)=\!\frac{g(t)}{t^3(t+18)^2},
  \]
  where
  \[
  g(t)= 2(t+18)^2\ln \left(\frac{t+18}{8.5}\right) -(3t^2+36t),\,\; t\geq0.
  \]
  Hence, to show $f''(t)> 0$, it sufficient to show that $g(t)>0$ for $t>0$.

  By some direct calculations, we have
  \begin{align*}
  g'(t)=\ & 4(t+18)\ln \left(\frac{t+18}{8.5}\right) -4t,\\
  g''(t)=\ & 4\ln \left(\frac{t+18}{8.5}\right).
  \end{align*}
  Clearly, $g''(t)>0$ when $t\geq0$, hence $g'(t)$ is increasing with $t$.
  Furthermore, $g'(0)>0$, therefore $g'(t)>0$ for $t\geq 0$.
  By some direct calculations,  $g(0)>0$.
  Therefore, $g(t)>0$ for $t> 0$.
  \end{proof}

  \section{Proof of Corollary~\ref{c:SC}}
  \label{ss:Pc}
  \begin{proof}
  Since
  \begin{equation}
  \label{e:con1}
  \frac{2^{3\ln2}}{17^{2\ln2}}<0.083216,
  \end{equation}
  by \eqref{e:SC}, to show \eqref{e:SC2}, we show
  \begin{equation}nn
  \left (  \frac{4x-1}{17} \right )^\frac{1}{2x-1}<1.04521<\frac{0.08698}{0.083216}, \mbox{ for integer } x\geq 5.
  \end{equation}nn
  Let
  \[
  f(x)=\left (  \frac{4x-1}{17} \right )^\frac{1}{2x-1},\,\; x\geq 5,
  \]
  then
  \[
  f(12)=1.045206...<1.04521.
  \]
  Hence, to show Corollary \ref{c:SC}, we only need to show that
  \begin{equation}
  \label{c:SCineq}
  f(x)\leq f(12), \mbox{ for any integer } x\geq 5.
  \end{equation}

  To this end, let $g(x)=\ln(f(x))$, then
  \[
  g'(x)=\frac{(1-4x)\ln\frac{4x-1}{17}+4x-2}{(2x-1)^2(4x-1)}.
  \]
  Let
  \[
  h(x)=(1-4x)\ln\frac{4x-1}{17}+4x-2, \,\; x\geq 5
  \]
  then
  \[
  h'(x)=-4\ln\frac{4x-1}{17}<0.
  \]
  Hence, $h(x)$ is  monotonically decreasing.
  Furthermore, by some direct calculations, one can show that $h(11)>0, h(12)<0$,
  therefore, $h(x)>0$ for $x\leq 11$ and  $h(x)<0$ for $x\geq 12$.
  Hence $g'(x)>0$ for $x\leq 11$ and  $g'(x)<0$ for $x\geq 12$.
  Since $g(x)=\ln(f(x))$, $f(x)$ is increasing for $x\leq 11$ and decreasing for $x\geq 12$,
  therefore \eqref{c:SCineq} holds.
  \end{proof}

  \section{Proof of Lemma~\ref{l:rie}}
  \label{ss:Prie}

  \begin{proof}
  By the definition of $\zeta(s)$, we have
  \[
  \zeta(2)< 1+\sum_{n=2}^{\infty}\frac{1}{n(n-1)}
  =1+\sum_{n=2}^{\infty}\left(\frac{1}{n-1}-\frac{1}{n}\right)<2.
  \]
  Furthermore, for $i\geq 3$, we have
  \[
  \zeta(i)\leq \zeta(3)<\frac{\pi^2}{8},
  \]
  where the second inequality follows from \cite{Eul73}. Then, we have
  \[
  \prod_{i=2}^{k}\zeta(i)\leq 2 \Big(\frac{\pi^2}{8}\Big)^{k-2}=\pi^{2(k-2)}2^{7-3k}.
  \]
  Since $\zeta(i)>1$ for $i\geq k+1$, we have
  \[
  \frac{\prod_{i=k+1}^{2k}\zeta(i)}{\prod_{i=2}^{k}\zeta(i)}
  > \frac{1}{\prod_{i=2}^{k}\zeta(i)}
  \geq \pi^{4-2k}2^{3k-7}.
  \]
  \end{proof}

  \section{Proof of Lemma~\ref{l:gamma}}
  \label{ss:Pgamma}

  \begin{proof}
  By \cite[Theorem 1.6]{Bat08}, for $x\geq 1$, we have
  \[
  (x/e)^x\sqrt{2\pi x}<\Gamma(x+1)<(x/e)^x\sqrt{2\pi (x+1/2)}.
  \]

  Therefore,
  \begin{align}
  \label{e:gamma1}
  &\,\frac{\prod_{i=k+1}^{2k}\Gamma(\frac{i}{2}+1)}{\prod_{i=2}^{k}\Gamma(\frac{i}{2}+1)}\nonumber\\
  > \ &\frac{\prod_{i=k+1}^{2k}(i/2e)^{i/2}\sqrt{\pi i}}{\prod_{i=2}^{k}(i/2e)^{i/2}\sqrt{\pi (i+1)}}\nonumber\\
  =\ &\sqrt{\pi}\sqrt{\frac{2(2k)!}{k!(k+1)!}}(2e)^{\frac{1}{2}(\sum_{i=2}^ki-\sum_{i=k+1}^{2k}i)}
  \sqrt{\frac{\prod_{i=k+1}^{2k}i^{i}}{\prod_{i=2}^{k}i^{i}}}\nonumber\\
  =\ &\sqrt{\pi}\sqrt{\frac{2(2k)!}{k!(k+1)!}}(2e)^{-\frac{k^2+1}{2}}
  \sqrt{\frac{\prod_{i=k+1}^{2k}i^{i}}{\prod_{i=2}^{k}i^{i}}}\nonumber\\
  =\ &\sqrt{\pi}2^{-\frac{k^2}{2}}e^{-\frac{k^2+1}{2}}\sqrt{\frac{(2k)!}{k!(k+1)!}}
  \sqrt{\frac{\prod_{i=k+1}^{2k}i^{i}}{\prod_{i=2}^{k}i^{i}}}.
  \end{align}

  In the following, we give a lower bound on the last term of the right-hand side of \eqref{e:gamma1}.
  Since $x\ln x$ is an increasing function for $x\geq \frac{1}{e}$, we have
  \begin{align*}
  &\,\ln\left(\frac{\prod_{i=k+1}^{2k}i^{i}}{\prod_{i=2}^{k}i^{i}}\right)\\
  =\ &\sum_{i=k+1}^{2k}i\ln i-\sum_{i=2}^ki\ln i\\
  \geq&\sum_{i=k+1}^{2k}\int_{i-1}^i x\ln xdx-\sum_{i=2}^k\int_{i}^{i+1}x\ln xdx\\
  =\ &\int_{k}^{2k} x\ln xdx-\int_{2}^{k+1}x\ln x dx\\
  =\ & \Big[2k^2\ln(2k)-k^2-\frac{k^2\ln k}{2}+\frac{k^2}{4}\Big]\\
  & -\Big[\frac{(k+1)^2\ln (k+1)}{2}-\frac{(k+1)^2}{4}-2\ln 2+1\Big]\\
  \overset{(a)}{>}& \Big(2k^2\ln k+2k^2\ln2-\frac{k^2\ln k}{2}-\frac{3k^2}{4}\Big)\\
  &-\Big(\frac{(k^2+2k+1)\ln k}{2}+1-\frac{k^2}{4}\Big)\\
  =\ &\frac{(2k^2-2k-1)}{2}\ln k+2k^2\ln2-\frac{k^2+2}{2},
  \end{align*}
  where (a) follows from the following inequalities,
  \begin{align*}
  &\frac{(k+1)^2\ln (k+1)}{2}-\frac{(k+1)^2}{4}-2\ln 2+1\\
  < \ &\frac{(k+1)^2\ln k}{2}+\frac{(k+1)^2\ln (1+\frac{1}{k})}{2}-\frac{(k+1)^2}{4}\\
  < \ &\frac{(k^2+2k+1)\ln k}{2}+\frac{(k^2+2k+1)\frac{1}{k}}{2}-\frac{(k+1)^2}{4}\\
  = \ &\frac{(k^2+2k+1)\ln k}{2}+\frac{k}{2}+1+\frac{1}{2k}-\frac{k^2+2k+1}{4}\\
  \leq \ &\frac{(k^2+2k+1)\ln k}{2}+1-\frac{k^2}{4},
  \end{align*}
  where the last inequality is from $k\geq 2$.
  Hence, we have
  \[
  \sqrt{\frac{\prod_{i=k+1}^{2n}i^{i}}{\prod_{i=2}^{k}i^{i}}}
  \geq 2^{k^2}e^{-\frac{k^2+2}{4}}k^{\frac{2k^2-2k-1}{4}},
  \]
  which combines \eqref{e:gamma1} imply \eqref{e:gamma}.
  \end{proof}

\end{appendices}


\bibliography{sn-bibliography}% common bib file


\end{document}